\documentclass[a4paper]{amsart}
\usepackage{geometry}
\usepackage[utf8]{inputenc}
\usepackage[T1]{fontenc}
\usepackage{lmodern}
\usepackage{latexsym}
\usepackage{amsmath}
\usepackage{amsfonts}
\usepackage{amssymb}
\usepackage{amsthm}
\usepackage[color=Olive]{todonotes} 
\usepackage{pgfplots,tikz}
\pgfplotsset{compat=1.15}
\usepackage{xspace} %
\usepackage{enumerate}%
\usetikzlibrary{arrows}
\usetikzlibrary{patterns}
\usetikzlibrary{arrows}
\usetikzlibrary{patterns}
\usetikzlibrary{shapes,decorations}
\usetikzlibrary{decorations.markings}
\usetikzlibrary{calc}
\usepackage{lineno}
\usepackage{hyperref}
\hypersetup{
colorlinks=true,
linkcolor=black,
citecolor=black,
filecolor=black,
urlcolor=black,
pdfproducer={LaTeX},
pdfcreator={pdfLaTeX}
}

\theoremstyle{plain}
\newtheorem{theorem}{Theorem}

\newtheorem{lemma}[theorem]{Lemma}
\newtheorem{claim}{Claim}[theorem]
\newtheorem{subclaim}{Subclaim}[theorem]

\let\leq\leqslant
\let\geq\geqslant

\newcommand{\yes}{{\bf{yes}}}
\newcommand{\no}{{\bf{no}}}
\newcommand{\bpd}{{\sc{Bipartite Permutation Vertex Deletion}}}

\author[J. Derbisz]{Jan Derbisz$^1$}

\title[A polynomial kernel for BPVD]{A polynomial kernel for vertex deletion into bipartite permutation graphs}

\address{$^1$Theoretical Computer Science Department, 
Faculty of Mathematics and Computer Science, Jagiellonian University in Krak\'ow, Poland.}
\email{jan.derbisz@doctoral.uj.edu.pl}

\begin{document}
\maketitle

\begin{abstract}

A permutation graph can be defined as an intersection graph of segments whose endpoints lie on two parallel lines $\ell_1$ and $\ell_2$, one on each. A bipartite permutation graph is a permutation graph which is bipartite.

In the the bipartite permutation vertex deletion problem we ask for a given $n$-vertex graph, whether we can remove at most $k$ vertices to obtain a bipartite permutation graph. This problem is NP-complete but it does admit an FPT algorithm parameterized by $k$.

In this paper we study the kernelization of this problem and show that it admits a polynomial kernel with $O(k^{62})$ vertices.

\end{abstract}

\section{Introduction}
In the vertex deletion problem into a graph class $\mathcal{G}$ we are given on input a graph $G=(V,E)$ and a natural number $k$. The goal is to determine whether there exists a subseteq $V'$ of $V$ such that $|V'|\leq k$ and $G-V'$ belongs to $\mathcal{G}$. It has been shown by Lewis and Yannakakis that this problem is $NP$-hard when the graph class $\mathcal{G}$ is nontrivial (i.e. there are infinitely many graphs in $\mathcal{G}$ and infinitely many graphs not in $\mathcal{G}$) and hereditary (i.e. closed under taking induced subgraphs).

This motivates the research of these problems from the point of view of parameterized complexity. Many of them have been shown to be in FPT parameterized by $k$ (i.e. admit algorithms with time complexity $\mathcal{O}(poly(|V|)\cdot f(k))$ for a computable function $f$), for example in the cases of chordal graphs \cite{Marx06}, interval graphs \cite{CaoMarxInt14}, proper interval graphs \cite{Cao15}, bipartite graphs \cite{DBLP:journals/orl/ReedSV04}, bipartite permutation graphs \cite{DBLP:journals/algorithmica/BozykDKNO22}. It is known that when a problem is in FPT, we can transform the input instance $(G, k)$ in polynomial time into an equivalent instance $(G', k')$  (called a kernel) with its total size bounded by a function $g(k)$. When $g$ is a polynomial function we say that the kernel is polynomial. A natural question to ask is whether a problem admits a polynomial kernel. Polynomial kernels have been found for vertex deletion problems into several graph classes, such as chordal graphs \cite{DBLP:journals/siamdm/JansenP18}, interval graphs \cite{DBLP:conf/soda/AgrawalM0Z19}, proper interval graphs \cite{DBLP:journals/siamdm/FominSV13}. A randomized polynomial kernel has also been found for the case of bipartite graphs \cite{DBLP:journals/talg/KratschW14}.

This paper is focused on the bipartite permutation graph class. A permutation graph can be defined as an intersection graph of segments whose endpoints lie on two parallel lines $\ell_1$ and $\ell_2$, one on each. A bipartite permutation graph is a permutation graph which is bipartite. It has been shown that vertex deletion into bipartite permutation graphs is in FPT \cite{DBLP:journals/algorithmica/BozykDKNO22}. The FPT algorithm uses the characterization of bipartite permutation graphs in terms of forbidden induced subgraphs in order to define an almost bipartite permutation graph, which is defined to exclude forbidden structures of bounded size but may contain larger ones. The algorithm uses branching technique for small forbidden structures in order to transform the input graph into an almost bipartite permutation graph. It is then possible to design a polynomial time algorithm which finds the smallest subset of vertices removal of which leaves us with a bipartite permutation graph.

The characterization of a graph class in terms of forbidden structures combined with the sunflower lemma has been successfully used by Fomin, Saurabh and Villanger in kernelization of vertex deletion into Proper Interval Graphs \cite{DBLP:journals/siamdm/FominSV13}. We use similar ideas to prove the main theorem of this paper.
\begin{theorem}
\label{theo:kernel}
Vertex deletion into bipartite permutation graphs admits a polynomial kernel with at most $O(k^{62})$ vertices.
\end{theorem}

\section{Preliminaries}

Unless stated otherwise, all graphs considered in this work are simple, i.e. undirected, with no loops and parallel edges.
Let $G=(V,E)$ be a graph.
We also denote the vertex set and the edge set as $V(G)$ and $E(G)$, respectively.
For a subset $S \subseteq V$, the subgraph of $G$ induced by $S$ is the graph $G[S] = (S, \{uv \mid uv \in E, u,v\in S\})$.
We write $G-v = G[V\setminus \{v\}]$ for a vertex $v\in V$.
Similarly, we write $G-S = G[V\setminus S]$ for a set  $S\subseteq V$.
The \emph{neighborhood} of a vertex $u \in V$ is the set $N(u) = \{v \in V \mid uv \in E\}$.
For a vertex $v\in V$ we write $N[v] = N(v)\cup \{v\}$.
Similarly, for a set $U\subseteq V$ we write $N(U)=\bigcup_{u\in U} N(u)\setminus U$ and $N[U]=\bigcup_{u\in U} N(u)\cup U$.
For a vertex $v\in V$ and a subset $S \subseteq V$ we write $N_S(v) = N(v)\cap S$.
We say that a vertex $v$ is \emph{isolated} in $G$ if $N(v)=\emptyset$. 
Let $u,v\in V$. 
We say that $u$ and $v$ \emph{are at distance~$k$ (in $G$)} if $k$ is the length of a shortest path between $u$ and $v$ in~$G$.
We write this as $dist_G(u,v)=k$.
We say that a vertex $v\in V$ and a set $S \subseteq V$ \emph{are at distance~$k$ (in $G$)} if $k$ is the minimum possible length of a shortest path between $v$ and some $u\in S$ in~$G$.

We denote a complete graph and a cycle on $n$ vertices by $K_n$ and $C_n$, respectively.
By \emph{hole} we mean an induced cycle on at least five vertices.
By a $5^+$-\emph{hole} we mean an induced cycle on at least five vertices.
A subset $I$ of $V(G)$ is called a \emph{component} of $G$ if $G[I]$ is a maximal connected (there is a path from any vertex of $G[I]$ to any other vertex of $G[I]$) induced subgraph of~$G$.
We say that a set $S\subseteq V-\{u,v\}$ is a $(u,v)$-separator or a vertex cut if $u\in V$ and $v\in V$ are in different connected components of $G-S$.

A \emph{partially ordered set} (shortly \emph{partial order} or \emph{poset}) is a pair $P = (X,{\leq_P})$ that consists of a set $X$ and a reflexive, transitive, and antisymmetric relation ${\leq_P}$ on $X$.
For a poset $(X,{\leq_P})$, let the \emph{strict partial order} $<_P$ be a binary relation defined on $X$ such that $x <_P y$ if and only if $x \leq_P y$ and $x \neq y$. Equivalently, $(X,{<_P})$ is a strict partial order if $<_P$ is irreflexive and transitive.
Two elements $x,y \in X$ are \emph{comparable} in $P$ if $x \leq_P y$ or $y \leq_P x$; otherwise, $x,y$ are \emph{incomparable} in $P$.
A \emph{linear order} $L=(X,{\leq_L})$ is a partial order in which for every $x,y \in X$ we have $x \leq_L y$ or $y \leq_L x$.
A \emph{strict linear order} $(X,<_L)$ is a binary relation defined in a way that $x <_L y$ if and only if $x \leq_L y$ and $x \neq y$. 

We say that two sets $X$ and $Y$ are \emph{comparable} if $X$ and $Y$ are comparable with respect to $\subseteq$-relation (that is, $X \subseteq Y$ or $Y \subseteq X$ holds).
We use the convenient notation $[m]:=\{1,\ldots,m\},$ for every $m \in \mathbb{N}$.
For every $i,j \in \mathbb{Z}$ such that $i \leq j$ by $[i,j]$ we mean the set $\{i,i+1,\ldots,j\}$.

We say that a set $X$ \emph{hits} set $Y$ if $X\cap Y\neq \emptyset$. 
We say that a set $X$ \emph{hits} a family of sets $\mathcal{F}$ when $X$ hits every set $Y\in \mathcal{F}$.


\subsection{The structure of bipartite permutation graphs}\label{sec_structure}
\label{sec:locally_bipartite_permutation_graphs}

The characterization of bipartite permutation graphs presented below was proposed by
Spinrad, Brandst\"{a}dt, and Stewart~\cite{SBS87}.

Suppose $G=(U,W,E)$ is a connected bipartite graph.
A linear order $(W,{<_W})$ satisfies \emph{adjacency property} 
if for each vertex $u \in U$ the set $N(u)$ consists of vertices that are consecutive in $(W,{<_W})$.
A linear order $(W,{<_W})$ satisfies \emph{enclosure property}
if for every pair of vertices $u,u' \in U$ such that $N(u)$ is a subset of $N(u')$, vertices in $N(u') - N(u)$ occur consecutively in $(W,{<_W})$.
A \emph{strong ordering} of the vertices of $U \cup W$ consists of 
linear orders $(U,{<_U})$ and $(W,{<_W})$ such that for every
$(u,w'), (u',w)$ in $E$, where $u,u'$ are in $U$ and $w,w'$ are in $W$,
$u <_U u'$ and $w <_W w'$ imply $(u,w) \in E$ and $(u',w') \in E$.
Note that, whenever $(U,{<_U})$ and $(W,{<_W})$ form a strong ordering of $U \cup W$, 
then $(U,{<_U})$ and $(W,{<_W})$ satisfy the adjacency and enclosure properties.
\begin{theorem}[Spinrad, Brandst\"{a}dt, Stewart~\cite{SBS87}]\label{thm:bip_char} The following three statements are equivalent for a connected bipartite graph $G=(U,W,E)$:
\begin{enumerate}
\item \label{thm:bip_char_1} $(U,W,E)$ is a bipartite permutation graph.
\item \label{thm:bip_char_2} There exists a strong ordering of $U \cup W$.
\item \label{thm:bip_char_3} There exists a linear order $(W,{<_W})$ of $W$ satisfying adjacency and enclosure properties.
\end{enumerate}
\end{theorem}
An example of a bipartite permutation graph $G=(U,W,E)$ with linear order $w_1 <_W w_2 <_W\ldots <_W w_8 <_W w_9$ of the vertices of $W$ which satisfies the adjacency and the enclosure properties is shown in Figure \ref{fig:bipartite_permutation_graph_kernel}.

Another characterization of bipartite permutation graphs can be obtained by listing all minimal forbidden induced subgraphs for this class of graphs.
Such a list can be compiled by taking all odd cycles of length $\geq 3$ (forbidden structures for bipartite graphs) and all bipartite graphs from the list of forbidden structures for permutation graphs obtained by Gallai~\cite{Gal67}. 
The whole list is shown in Figure~\ref{fig:bp_forbidden_structures}.
\begin{figure}[htp!]
\centering
\begin{tikzpicture}[xscale=0.8,yscale=1]
\coordinate (x1) at (0,0) {};
\coordinate (x2) at (1,0) {};
\coordinate (x3) at (2,0) {};
\coordinate (x4) at (3,0) {};
\coordinate (y1) at (0.5,1) {};
\coordinate (y2) at (1.5,1) {};
\coordinate (y3) at (2.25,1) {};
\coordinate (l) at (1.5,-0.5) {};

\tikzstyle{every node}=[circle,minimum size=5pt,inner sep=0pt,draw,fill]
\node at (x1) {};
\node at (x2) {};
\node at (x3) {};
\node at (x4) {};

\node at (y1) {};
\node at (y2) {};
\node at (y3) {};

\tikzstyle{every node}=[inner sep=1pt]

\path (y1) edge[thick] (x1);
\path (y1) edge[thick] (x2);

\path (y2) edge[thick] (x2);
\path (y2) edge[thick] (x3);

\path (y3) edge[thick] (x2);
\path (y3) edge[thick] (x4);

\begin{footnotesize}
\tikzstyle{every node}=[inner sep=2pt]
\node at (l) {$T_2$};
\end{footnotesize}
\end{tikzpicture}
\hspace{0.65cm}
\begin{tikzpicture}[xscale=0.8,yscale=1]
\coordinate (x1) at (0,0) {};
\coordinate (x2) at (1,0) {};
\coordinate (y1) at (0,1) {};
\coordinate (y2) at (1,1) {};
\coordinate (ay1) at (-1,0) {};
\coordinate (ax1) at (-1,1) {};
\coordinate (ax2) at (2,1) {};
\coordinate (l) at (0.5,-0.5) {};

\tikzstyle{every node}=[circle,minimum size=5pt,inner sep=0pt,draw,fill]
\node at (x1) {};
\node at (x2) {};
\node at (y1) {};
\node at (y2) {};

\node at (ax1) {};
\node at (ax2) {};
\node at (ay1) {};

\tikzstyle{every node}=[inner sep=1pt]

\path (y1) edge[thick] (x1);
\path (y1) edge[thick] (x2);
\path (y2) edge[thick] (x1);
\path (y2) edge[thick] (x2);

\path (x1) edge[thick] (ax1);
\path (x2) edge[thick] (ax2);
\path (y1) edge[thick] (ay1);

\begin{footnotesize}
\tikzstyle{every node}=[inner sep=2pt]
\node at (l) {$X_2$};
\end{footnotesize}
\end{tikzpicture}
\hspace{0.65cm}
\begin{tikzpicture}[xscale=0.8,yscale=1]
\coordinate (x1) at (0,0) {};
\coordinate (x2) at (1,0) {};
\coordinate (x3) at (2.5,0) {};
\coordinate (y1) at (0,1) {};
\coordinate (y2) at (1,1) {};
\coordinate (y3) at (2.5,1) {};
\coordinate (ax2) at (1.75,1) {};
\coordinate (l) at (1.25,-0.5) {};

\tikzstyle{every node}=[circle,minimum size=5pt,inner sep=0pt,draw,fill]
\node at (x1) {};
\node at (x2) {};
\node at (x3) {};
\node at (y1) {};
\node at (y2) {};
\node at (y3) {};
\node at (ax2) {};

\tikzstyle{every node}=[inner sep=1pt]

\path (x1) edge[thick] (y1);
\path (x1) edge[thick] (y2);
\path (x2) edge[thick] (y1);
\path (x2) edge[thick] (y2);
\path (x2) edge[thick] (y3);
\path (x3) edge[thick] (y2);
\path (x3) edge[thick] (y3);

\path (x2) edge[thick] (ax2);

\begin{footnotesize}
\tikzstyle{every node}=[inner sep=2pt]
\node at (l) {$X_3$};
\end{footnotesize}
\end{tikzpicture}
\hspace{0.65cm}
\begin{tikzpicture}[xscale=0.8,yscale=1]
\coordinate (x1) at (-0.25,0) {};
\coordinate (x2) at (1,0) {};
\coordinate (x3) at (2.25,0) {};
\coordinate (y1) at (-0.25,1) {};
\coordinate (y2) at (1,1) {};
\coordinate (y3) at (2.25,1) {};
\coordinate (l) at (1,-0.5) {};

\tikzstyle{every node}=[circle,minimum size=5pt,inner sep=0pt,draw,fill]
\node at (x1) {};
\node at (x2) {};
\node at (x3) {};
\node at (y1) {};
\node at (y2) {};
\node at (y3) {};

\tikzstyle{every node}=[inner sep=1pt]

\path (x1) edge[thick] (y3);
\path (x1) edge[thick] (y1);
\path (x2) edge[thick] (y1);
\path (x2) edge[thick] (y2);
\path (x3) edge[thick] (y2);
\path (x3) edge[thick] (y3);

\begin{footnotesize}
\tikzstyle{every node}=[inner sep=2pt]
\node at (l) {$\text{$C_{2k}$ for $k \geq 3$}$};
\end{footnotesize}
\end{tikzpicture}

\vspace{0.2cm}
\begin{tikzpicture}[xscale=0.8,yscale=1]
\coordinate (x1) at (-0.5,0) {};
\coordinate (x2) at (1,0) {};
\coordinate (y1) at (0.25,1) {};
\coordinate (l) at (0.25,-0.5) {};

\tikzstyle{every node}=[circle,minimum size=5pt,inner sep=0pt,draw,fill]
\node at (x1) {};
\node at (x2) {};
\node at (y1) {};

\tikzstyle{every node}=[inner sep=1pt]

\path (x1) edge[thick] (x2);
\path (x1) edge[thick] (y1);
\path (x2) edge[thick] (y1);
\begin{footnotesize}
\tikzstyle{every node}=[inner sep=2pt]
\node at (l) {$K_3$};
\end{footnotesize}
\end{tikzpicture}
\hspace{1cm}
\begin{tikzpicture}[xscale=0.8,yscale=1]
\coordinate (x1) at (-0.25,0) {};
\coordinate (x2) at (1,0) {};
\coordinate (z1) at (2.25,0.5) {};
\coordinate (y1) at (-0.25,1) {};
\coordinate (y2) at (1,1) {};
\coordinate (l) at (1,-0.5) {};

\tikzstyle{every node}=[circle,minimum size=5pt,inner sep=0pt,draw,fill]
\node at (x1) {};
\node at (x2) {};
\node at (y1) {};
\node at (y2) {};
\node at (z1) {};

\tikzstyle{every node}=[inner sep=1pt]

\path (x1) edge[thick] (x2);
\path (x2) edge[thick] (z1);
\path (z1) edge[thick] (y2);
\path (y2) edge[thick] (y1);
\path (y1) edge[thick] (x1);

\begin{footnotesize}
\tikzstyle{every node}=[inner sep=2pt]
\node at (l) {$\text{$C_{2k+1}$ for $k \geq 2$}$};
\end{footnotesize}
\end{tikzpicture}
\caption{\label{fig:bp_forbidden_structures} 
Forbidden structures for bipartite permutation graphs.}
\end{figure}
\begin{figure}[htp!]
\centering

\begin{tikzpicture}[xscale=1.3,yscale=1.5]

\begin{axis}[height=2.6cm, width=12cm,
    hide axis,
    view = {0}{90},
    at={(-1.22cm,0)}
    ]
 \addplot3 [
    surf,
    colormap={blackwhite}{gray(0cm)=(1); gray(1cm)=(0.9)},
    shader     = faceted interp,
    point meta = x,
    samples    = 15,
    samples y  = 3,
    z buffer   = sort,
    domain     = -.5:8.5,
    y domain   = 0:1
    ] (
    {x},
    {y/4},
    {0}
    );
 \addplot3 [color=black,
    domain     = -.5:8.5,samples y=0,samples=2*(640/360)*24+1,
    ] (
    {x},
    {0},
    {0} 
    );
    
    \addplot3 [color=black,
    domain     = -.5:8.5,samples y=0,samples=2*(640/360)*24+1,
    ] (
    {x},
    {1/4},
    {0}
    );   
    
\end{axis}

\coordinate (w1) at (0,1) {};
\coordinate (w2) at (1,1) {};
\coordinate (w3) at (2,1) {};
\coordinate (w4) at (3,1) {};
\coordinate (w5) at (4,1) {};
\coordinate (w6) at (5,1) {};
\coordinate (w7) at (6,1) {};
\coordinate (w8) at (7,1) {};
\coordinate (w9) at (8,1) {};

\coordinate (lw1) at (0,1.33) {};
\coordinate (lw2) at (1,1.33) {};
\coordinate (lw3) at (2,1.33) {};
\coordinate (lw4) at (3,1.33) {};
\coordinate (lw5) at (4,1.33) {};
\coordinate (lw6) at (5,1.33) {};
\coordinate (lw7) at (6,1.33) {};
\coordinate (lw8) at (7,1.33) {};
\coordinate (lw9) at (8,1.33) {};

\coordinate (u1) at (0,0) {};
\coordinate (u2) at (1,0) {};
\coordinate (u3) at (2,0) {};
\coordinate (u4) at (3,0) {};
\coordinate (u5) at (4,0) {};
\coordinate (u6) at (5,0) {};
\coordinate (u7) at (6,0) {};
\coordinate (u8) at (7,0) {};
\coordinate (u9) at (8,0) {};

\coordinate (lu1) at (0,-0.33) {};
\coordinate (lu2) at (1,-0.33) {};
\coordinate (lu3) at (2,-0.33) {};
\coordinate (lu4) at (3,-0.33) {};
\coordinate (lu5) at (4,-0.33) {};
\coordinate (lu6) at (5,-0.33) {};
\coordinate (lu7) at (6,-0.33) {};
\coordinate (lu8) at (7,-0.33) {};
\coordinate (lu9) at (8,-0.33) {};

\begin{scope}[fill opacity=0.5]
\draw[rounded corners=7, fill=gray!30, thick] (3.7,0.8)--(3.7,1.2) -- (6.3,1.2) -- (6.3,0.8)--(5,-0.05)--cycle;
\draw[rounded corners=7, fill=gray!40, thick] (3.7,0.8)--(3.7,1.2) -- (6.3,1.2) -- (6.3,0.8)--cycle;
\draw[rounded corners=5, fill=gray!60, thick] (4.7,0.85)--(4.7,1.15) -- (6.2,1.15) -- (6.2,0.85)--(6.05,-0.05)--cycle;
\draw[rounded corners=5, fill=gray!70, thick] (4.7,0.85)--(4.7,1.15) -- (6.2,1.15) -- (6.2,0.85)--cycle;
\end{scope}


\path (u1) edge[thick] (w1);
\path (u1) edge[thick] (w2);
\path (u1) edge[thick] (w3);

\path (u2) edge[thick] (w2);
\path (u2) edge[thick] (w3);
\path (u2) edge[thick] (w4);

\path (u3) edge[thick] (w2);
\path (u3) edge[thick] (w3);
\path (u3) edge[thick] (w4);

\path (u4) edge[thick] (w3);
\path (u4) edge[thick] (w4);
\path (u4) edge[thick] (w5);
\path (u4) edge[thick] (w6);

\path (u5) edge[thick] (w4);
\path (u5) edge[thick] (w5);
\path (u5) edge[thick] (w6);

\path (u6) edge[thick] (w5);
\path (u6) edge[thick] (w6);
\path (u6) edge[thick] (w7);

\path (u7) edge[thick] (w6);
\path (u7) edge[thick] (w7);

\path (u8) edge[thick] (w7);
\path (u8) edge[thick] (w8);
\path (u8) edge[thick] (w9);

\path (u9) edge[thick] (w8);
\path (u9) edge[thick] (w9);

\tikzstyle{every node}=[circle,minimum size=5pt,inner sep=0pt,draw,fill]
\node at (w1) {};
\node at (w2) {};
\node at (w3) {};
\node at (w4) {};
\node at (w5) {};
\node at (w6) {};
\node at (w7) {};
\node at (w8) {};
\node at (w9) {};
\node at (u1) {};
\node at (u2) {};
\node at (u3) {};
\node at (u4) {};
\node at (u5) {};
\node at (u6) {};
\node at (u7) {};
\node at (u8) {};
\node at (u9) {};

\tikzstyle{every node}=[inner sep=2pt]
\node at (lw1) {$w_1$};
\node at (lw2) {$w_2$};
\node at (lw3) {$w_3$};
\node at (lw4) {$w_4$};
\node at (lw5) {$w_5$};
\node at (lw6) {$w_6$};
\node at (lw7) {$w_7$};
\node at (lw8) {$w_8$};
\node at (lw9) {$w_9$};

\node at (lu1) {$u_1$};
\node at (lu2) {$u_2$};
\node at (lu3) {$u_3$};
\node at (lu4) {$u_4$};
\node at (lu5) {$u_5$};
\node at (lu6) {$u_6$};
\node at (lu7) {$u_7$};
\node at (lu8) {$u_8$};
\node at (lu9) {$u_9$};

\end{tikzpicture}

\caption{\label{fig:bipartite_permutation_graph_kernel}
Embedding of a bipartite permutation graph $(U,W,E)$ into a strip satysfying the adjacency and the enclosure properties.
}
\end{figure}

\section{Sunflower lemma}
We say that a family of $m$ sets $\{S_1, S_2, \ldots , S_m\}$ over a universe $\mathcal{U}$ forms a \emph{sunflower} with \emph{core} $Y$ and $m$ \emph{petals} $S_i \setminus Y$ for $i\in [m]$ when for every $i,j \in [m]$ and $i\neq j$ we have $S_i \cap S_j = Y$ and for every $i\in [m]$ we have $S_i \setminus Y \neq \emptyset$.

The sunflower lemma is stated as follows:
\begin{lemma}[Sunflower lemma \cite{DBLP:series/txtcs/FlumG06}]
\label{lem:sl}
Let $\mathcal{F}$ be a family of sets over a universe $\mathcal{U}$ each of
cardinality at most $d$. If $|\mathcal{F}| > d!(k - 1)^d$
then $\mathcal{F}$ contains a sunflower with $k$ petals and such
a sunflower can be found in $O(k + |\mathcal{F}|)$ time.
\end{lemma}

The sunflower lemma has been successfully used by Fomin, Saurabh and Villanger in kernelization of vertex deletion into proper interval graphs~\cite{DBLP:journals/siamdm/FominSV13}. 
We use analogous ideas combined with the result from~\cite{DBLP:journals/algorithmica/BozykDKNO22}
for bipartite permutation graphs.

Recall that bipartite permutation graphs can be characterized by forbidden induced subgraphs (see Figure~\ref{fig:bp_forbidden_structures}). In particular, given an instance $(G, k)$ our task is to determine if there is a set $X\subseteq V(G)$ hitting all forbidden subgraphs such that $|X|\leq k$. In this paper, we say that a forbidden structure is \emph{small} if it has at most 12 vertices. Otherwise we say the structure is \emph{big} or \emph{long}. Note that every big structure is a $5^+$-hole.

Fomin, Saurabh and Villanger \cite{DBLP:journals/siamdm/FominSV13} used the sunflower lemma to show the following lemma:

\begin{lemma}[\cite{DBLP:journals/siamdm/FominSV13}]
\label{lem:hs}
Let $\mathcal{F}$ be a family of sets of cardinality at most $d$ over a universe $\mathcal{U}$ and $k$ be a
positive integer. Then there is an $O(|\mathcal{F}|(k + |\mathcal{F}|))$ time algorithm that finds a non-empty set
$\mathcal{F}' \subseteq \mathcal{F}$ such that
\begin{itemize}
    \item For every $Z \subseteq \mathcal{U}$ of size at most $k$, $Z$ is a minimal hitting set of $\mathcal{F}$ if and only if $Z$ is a minimal hitting set of $\mathcal{F}'$,
\item $|\mathcal{F'}
| \leq d!(k + 1)^d$.

\end{itemize}

\end{lemma}

The algorithm applies the sunflower lemma to progressively remove sets from $\mathcal{F}$ until it has desired size. If $|\mathcal{F}| \leq d!(k + 1)^d$ it returns $\mathcal{F'} = \mathcal{F}$. Otherwise, due to the sunflower lemma, there is a sunflower with core $Y$ and $k+2$ petals $\{S_1, \ldots, S_{k+2}\}$ in $\mathcal{F}$. The set $S_{k+2}$ is then removed from $\mathcal{F}$. The algorithm continues until $|\mathcal{F}| \leq d!(k + 1)^d$. Note that whenever a set $X$ such that $|X|\leq k$ hits every $S_i$ for $i\in [k+2]$, then $X$ must also hit the core $Y$.

Let us define $\delta(k) = 12\cdot 12!(k+1)^{12}+k = O(k^{12})$.
Similarly to \cite{DBLP:journals/siamdm/FominSV13}, we can then prove the following lemma.

\begin{lemma}
\label{lem:sunflower}
Let $(G, k)$ be an instance of \bpd.
In polynomial time we can conclude that $(G, k)$ is a \no-instance or find a set $T\subseteq V(G)$ such that:
\begin{itemize}
    \item $|T|\leq \delta(k)$,
    \item $G-T$ is a bipartite permutation graph,
    \item for every set $X\subseteq V(G)$ such that $|X|\leq k$, $X$ is a minimal set hitting small forbidden induces subgraphs in $G$ iff $X$ is a minimal set hitting small forbidden structures which are contained in $G[T]$.
\end{itemize}
\end{lemma}
\begin{proof}
Let $\mathcal{U}=V(G)$ and let $\mathcal{F}$ be the family of all small forbidden structures of~$G$. Note that $|\mathcal{F}| \leq |V(G)|^{12}$. We apply Lemma~\ref{lem:hs} to $\mathcal{F}$ and $k$ to compute the family $\mathcal{F}' \subseteq \mathcal{F}$ with properties stated in Lemma~\ref{lem:hs}. Let $S\subseteq V(G)$ be the set of vertices belonging to any set $Y\in \mathcal{F}'$ (i.e. $S = \bigcup \mathcal{F}'$). Due to the properties of $\mathcal{F}'$ from Lemma~\ref{lem:hs} we have that $|S|\leq 12\cdot 12!(k+1)^{12}$ and $G-S$ does not include any small forbidden subgraphs. As shown in~\cite{DBLP:journals/algorithmica/BozykDKNO22}, it is possible to find in polynomial time a minimum size set $X\subseteq V(G)-S$ such that $(G-S)-X$ is a bipartite permutation graph, due to $G-S$ being an almost bipartite permutation graph. If $|X|>k$, we conclude that $(G, k)$ is a \no-instance. Otherwise we define $T=S\cup X$. Since  $|S|\leq 12\cdot 12!(k+1)^{12}$ and $|X|\leq k$ we have that $|T| \leq 12\cdot 12!(k+1)^{12}+k=\delta(k)$, as desired.
\end{proof}

\section{Overview of the kernelization algorithm}
In this section we give an overview of the kernelization algorithm for \bpd.

We will introduce a set of reduction rules such that for input instance $(G, k)$:
\begin{itemize}
    \item either $|V(G)|$ is upper-bounded by a polynomial $\xi(k)=O(k^{62})$, in which case we return $(G, k)$ as the kernel,
    \item or we can apply one of the reduction rules to find in polynomial time an equivalent instance $(G', k')$ with $|V(G')|<|V(G)|$ and $k'\leq k$.
\end{itemize}
Therefore, by repeatedly applying reduction rules, the kernelization algorithm runs in total in polynomial time and returns a kernel of size upper-bounded by $\xi(k)=O(k^{62})$, as desired.

 We first compute the set $T$ as described in previous section (or conclude that $(G, k)$ is a \no-instance). Then, we are going to investigate the properties of the graph $G-T$ (see Figure~\ref{fig:sets}).
 \begin{figure}[htp!]
\centering
\begin{tikzpicture}[xscale=1.2,yscale=1.2]

\draw[red] (0, 0) ellipse (2cm and 1 cm);
\node[red] at (0,0) {$T$};
\coordinate (a) at (-1, 1.3);
\coordinate (b) at (0, 1.3);
\coordinate (c) at (1, 1.3);
\node[circle, black, fill, inner sep=1.8pt] at (a) {};
\node[circle, black, fill, inner sep=1.8pt] at (b) {};
\node[circle, black, fill, inner sep=1.8pt] at (c) {};

\path (a) edge ($(a)+(-0.2, -0.7)$);
\path (a) edge ($(a)+(0, -0.7)$);
\path (a) edge ($(a)+(0.2, -0.7)$);

\path (b) edge ($(b)+(-0.2, -0.7)$);
\path (b) edge ($(b)+(0, -0.7)$);
\path (b) edge ($(b)+(0.2, -0.7)$);

\path (c) edge ($(c)+(-0.2, -0.7)$);
\path (c) edge ($(c)+(0, -0.7)$);
\path (c) edge ($(c)+(0.2, -0.7)$);

\draw (0, -1.7) ellipse (1cm and 0.5 cm);
\node at (0,-1.7) {$C$};

\path (-0.6, -1.5) edge (-1.2, -0.6);
\path (-0.35, -1.4) edge (-0.8, -0.6);
\path (-0.1, -1.4) edge (-0.4, -0.6);

\path (0.6, -1.5) edge (1.2, -0.6);
\path (0.35, -1.4) edge (0.8, -0.6);
\path (0.1, -1.4) edge (0.4, -0.6);

\end{tikzpicture}

\caption{
\label{fig:sets}
The set $T$ and the graph $G-T$. The set $C$ contains non-isolated vertices in $G-T$.
}

\end{figure}
 We will be using the properties of $T$ and the irrelevant vertex technique in order to obtain reduction rules upper-bounding the number of isolated vertices in $G-T$ by a polynomial $O(k^{25})$. Then we introduce reduction rules ensuring that the set $C\subseteq V$ of non-isolated vertices in $G-T$ has size upper-bounded by a polynomial $O(k^{62})$. Since $|T|$ is upper-bounded by a polynomial $\delta(k)$, we conclude that the kernel has size upper-bounded by a polynomial $O(k^{62})$.

\section{Reduction rules}
In this paper, we say that $X \subseteq V$ is a \emph{small forbidden structure} if $G[X]$ is isomorphic to one of the graphs: $K_3, T_2, X_2, X_3, C_5, C_6, C_7, C_8, C_9, C_{10}, C_{11}, C_{12}$.

First, we prove the following lemma.
\begin{lemma}[The Path Lemma]
\label{lem:path}
Let $H$ be a bipartite permutation graph with bipartition sides $U$ and $W$ ordered by strong ordering $(<_U, <_W)$ provided by Theorem~\ref{thm:bip_char}. Let $P$ be an induced path in $H$. Then for every $u\in H-P$ we have $|N(u)\cap V(P)|\leq 3$ and consecutive vertices of this set are at distance 2 on $P$.
\end{lemma}
\begin{proof}
Suppose consecutive vertices of $N(u)\cap V(P)$ are not at distance 2 on $P$. Then there is $K_3$ or a $5^+$-hole in $H$, which contradicts that $H$ is a bipartite permutation graph.

Assume $u\in U$ (the case of $u\in W$ is symmetrical).
Suppose $|N(u)\cap V(P)| \geq 4$. Let $w_0 <_W w_2 <_W w_4 <_W w_6$ be consecutive neighbours of $u$ on $P$ and let $u_1 <_U u_3 <_U u_5$ be the vertices of $P$ between them (see Figure~\ref{fig:path}).
\begin{figure}[htp!]
\centering

\begin{tikzpicture}[xscale=1.3,yscale=1.5]

\begin{axis}[height=2.6cm, width=12cm,
    hide axis,
    view = {0}{90},
    at={(-1.22cm,0)}
    ]
 \addplot3 [
    surf,
    colormap={blackwhite}{gray(0cm)=(1); gray(1cm)=(0.9)},
    shader     = faceted interp,
    point meta = x,
    samples    = 15,
    samples y  = 3,
    z buffer   = sort,
    domain     = -.5:8.5,
    y domain   = 0:1
    ] (
    {x},
    {y/4},
    {0}
    );
 \addplot3 [color=black,
    domain     = -.5:8.5,samples y=0,samples=2*(640/360)*24+1,
    ] (
    {x},
    {0},
    {0} 
    );
    
    \addplot3 [color=black,
    domain     = -.5:8.5,samples y=0,samples=2*(640/360)*24+1,
    ] (
    {x},
    {1/4},
    {0}
    );   
    
\end{axis}

\coordinate (w1) at (0,1) {};
\coordinate (w2) at (1,1) {};
\coordinate (w3) at (2,1) {};
\coordinate (w4) at (3,1) {};
\coordinate (w5) at (4,1) {};
\coordinate (w6) at (5,1) {};
\coordinate (w7) at (6,1) {};
\coordinate (w8) at (7,1) {};
\coordinate (w9) at (8,1) {};

\coordinate (lw1) at (0,1.33) {};
\coordinate (lw2) at (1,1.33) {};
\coordinate (lw3) at (2,1.33) {};
\coordinate (lw4) at (3,1.33) {};
\coordinate (lw5) at (4,1.33) {};
\coordinate (lw6) at (5,1.33) {};
\coordinate (lw7) at (6,1.33) {};
\coordinate (lw8) at (7,1.33) {};
\coordinate (lw9) at (8,1.33) {};

\coordinate (u1) at (0,0) {};
\coordinate (u2) at (1,0) {};
\coordinate (u3) at (2,0) {};
\coordinate (u4) at (3,0) {};
\coordinate (u5) at (4,0) {};
\coordinate (u6) at (5,0) {};
\coordinate (u7) at (6,0) {};
\coordinate (u8) at (7,0) {};
\coordinate (u9) at (8,0) {};

\coordinate (lu1) at (0,-0.33) {};
\coordinate (lu2) at (1,-0.33) {};
\coordinate (lu3) at (2,-0.33) {};
\coordinate (lu4) at (3,-0.33) {};
\coordinate (lu5) at (4,-0.33) {};
\coordinate (lu6) at (5,-0.33) {};
\coordinate (lu7) at (6,-0.33) {};
\coordinate (lu8) at (7,-0.33) {};
\coordinate (lu9) at (8,-0.33) {};


\path (u3) edge[thick] (w2);
\path (u3) edge[thick] (w4);

\path (u5) edge[thick] (w4);
\path (u5) edge[thick] (w6);

\path (u7) edge[thick] (w6);
\path (u7) edge[thick] (w8);

\tikzstyle{every node}=[circle,minimum size=5pt,inner sep=0pt,draw,fill]
\node at (w2) {};
\node at (w4) {};
\node at (w6) {};
\node at (w8) {};
\node at (u3) {};
\node at (u5) {};
\node at (u7) {};

\tikzstyle{every node}=[inner sep=2pt]
\node at (lw2) {$w_0$};
\node at (lw4) {$w_2$};
\node at (lw6) {$w_4$};
\node at (lw8) {$w_6$};

\node at (lu3) {$u_1$};
\node at (lu5) {$u_3$};
\node at (lu7) {$u_5$};

\end{tikzpicture}

\caption{\label{fig:path}
An induced subpathpath of $P$ ordered by strong ordering. Note that there is no way to place a vertex $u$ such that $\{w_0, w_2, w_4, w_6\} \subseteq N(u)$.
}
\end{figure}
Since $w_0 \notin N(u_3)$, $u<_U u_3$ in the strong ordering of $H$. But since $w_6 \notin N(u_3)$, $u_3 <_U u$, which is a contradiction.
\end{proof}

Let $(G, k)$ be an instance of \bpd~and let $T$ be the set provided by Lemma~\ref{lem:sunflower}. Due to the properties of $T$ we can claim the following.
\begin{claim}
\label{claim:one}
Let $Y$ be a small forbidden structure in $G$ such that $Y\cap T=\{v\}$. Then $(G, k)$ is equivalent to $(G-v, k-1)$.
\end{claim}
\begin{proof}
We need to show that if a set $X\subseteq V(G)$ is a solution for $(G, k)$, then $v\in X$. Let $X'\subseteq X$ be a minimal subset of $X$ hitting all small forbidden subgraphs of~$G$. Due to the properties of the set~$T$, $X'$ is minimal inclusion wise set hitting all small forbidden subgraphs of $G[T]$. In particular, $X'\subseteq T$. Since $X'$ hits $Y$ and $X'\subseteq T$, we have $v\in X'\subseteq X$.
\end{proof}
Claim~\ref{claim:one} allows us to further assume that $|Y\cap T|\geq 2$ for every small forbidden subgraph~$Y$.

We say that a vertex $v\in V(G)$ is \emph{irrelevant} if instances $(G, k)$ and $(G-v, k)$ are equivalent.


\section{Bounding the number of isolated vertices in $G-T$}

In this section we give a reduction rule bounding the number of isolated vertices in $G-T$ by a polynomial $\delta(k)^2(k+1)$.
For this purpose, create a set of vertices $Z\subseteq G-T$ the following way. Iterate over every pair of vertices $\{u, w\}$ in $T$. If $|(N(u)\cap N(w))-T|\leq k+1$ then add $(N(u)\cap N(w))-T$ to $Z$. Otherwise choose any subset of $(N(u)\cap N(w))-T$ of size $k+1$ and add it to $Z$. Note that at this point $|Z|\leq \delta(k)^2(k+1)$.
 
 We now prove the following claim. 
 \begin{claim}
 \label{claim:single}
 Let $v$ be an isolated vertex in $G-T$. If $v \notin Z$, then $v$ is irrelevant.
\end{claim}

\begin{proof}
Assume that $(G-v, k)$ is a \yes-instance. Let $X$ be a solution for $(G-v, k)$. We claim that $X$ is a solution for $(G, k)$ as well. 
Suppose the contrary. 
Then there is a forbidden structure $S$ containing $v$ in $G-X$. 
Since $v \notin T$, $X$ hits all small forbidden structures in $G[T]$.
By Lemma~\ref{lem:sunflower}, $X$ hits all small forbidden structures in~$G$.
Since $S$ is not hit by $X$, $S$ is a big structure, which means that $S$ is a $5^+$-hole.

Since $N(v)\subseteq T$ we have that the neighbours of $v$ on $S$ are $u, w \in T$.
Since $v$ was not added to $Z$,
$u$ and $w$ have at least $k+1$ common neighbours in $Z$, and 
at least one of them, say $v'$, is not in $X$ as $|X| \leq k$.
Vertices $u$ and $w$ are joined by an induced path $(u, v_1, v_2, ..., v_l, w)$ in $S$, where $l\geq 10$. Since $S$ is an induced cycle and $\{u, w\} \subseteq N(v')$, we have that $v'\notin S$. Since $(G-v)-X$ is a bipartite permutation graph and $l\geq 10$ we deduce there is a $5^+$-hole in $(G-v)-X$ formed by some vertices from the set $(S-v) \cup \{v'\}$ due to the path lemma (see Figure~\ref{fig:cycle}). This contradicts the fact that $(G-v)-X$ is a bipartite permutation graph.
\begin{figure}[htp!]
\centering
\begin{tikzpicture}[xscale=1,yscale=1]
\coordinate (u) at (-1,0) {};
\coordinate (x1) at (-0.9,-0.5) {};
\coordinate (x2) at (-0.7,-1) {};
\coordinate (x3) at (-0.4,-1.5) {};
\coordinate (x4) at (-0.1,-1.9) {};
\coordinate (x5) at (0.4,-2.1) {};
\coordinate (x6) at (0.9,-2.1) {};
\coordinate (x7) at (1.4,-1.9) {};
\coordinate (x8) at (1.7,-1.5) {};
\coordinate (x9) at (2,-1) {};
\coordinate (x10) at (2.2,-0.5) {};
\coordinate (w) at (2.3,0) {};
\coordinate (v) at (0,1) {};
\coordinate (v') at (1.3,1) {};
\coordinate (vl) at (-0.3,1.2) {};
\coordinate (v'l) at (1.7,1.2) {};
\coordinate (ul) at (-1.3,0) {};
\coordinate (x1l) at (-1.2,-0.5) {};
\coordinate (x2l) at (-1,-1) {};
\coordinate (x3l) at (-0.7,-1.5) {};
\coordinate (x4l) at (-0.4,-1.9) {};
\coordinate (x5l) at (0.1,-2.3) {};
\coordinate (x6l) at (1.2,-2.3) {};
\coordinate (x7l) at (1.7,-1.9) {};
\coordinate (x8l) at (2,-1.5) {};
\coordinate (x9l) at (2.3,-1) {};
\coordinate (x10l) at (2.6,-0.5) {};
\coordinate (wl) at (2.6,0) {};

\node at (vl) {$v$};
\node at (v'l) {$v'$};
\node[red] at (ul) {$u$};
\node[red] at (wl) {$w$};
\node at (x1l) {$v_1$};
\node at (x2l) {$v_2$};
\node at (x3l) {$v_3$};
\node at (x4l) {$v_4$};
\node at (x5l) {$v_5$};
\node at (x6l) {$v_6$};
\node at (x7l) {$v_7$};
\node at (x8l) {$v_8$};
\node at (x9l) {$v_9$};
\node at (x10l) {$v_{10}$};

\path (u) edge[thick] (x1);
\path (x1) edge[thick] (x2);
\path (x2) edge[thick] (x3);
\path (x3) edge[thick] (x4);
\path (x4) edge[thick] (x5);
\path (x5) edge[thick] (x6);
\path (x6) edge[thick] (x7);
\path (x7) edge[thick] (x8);
\path (x8) edge[thick] (x9);
\path (x9) edge[thick] (x10);
\path (x10) edge[thick] (w);
\path (v) edge[thick] (w);
\path (v) edge[thick] (u);
\path (v') edge[thick] (w);
\path (v') edge[thick] (u);
\path (v') edge (x9);

\tikzstyle{every node}=[circle,minimum size=5pt,inner sep=0pt,draw,fill]
\node at (x1) {};
\node at (x2) {};
\node at (x3) {};
\node at (x4) {};
\node at (x5) {};
\node at (x6) {};
\node at (x7) {};
\node at (x8) {};
\node at (x9) {};
\node at (x10) {};
\node[red] at (u) {};
\node[red] at (w) {};
\node at (v) {};
\node at (v') {};

\end{tikzpicture}
\caption{\label{fig:cycle} 
Forbidden induced cycle $S$ and the vertex $v'$. The vertices $u,w\in T$ are depicted in red. Vertex $v$ is irrelevant because a vertex $v'$ would also form a forbidden induced cycle.}
\end{figure}
\end{proof}

Therefore, we can further assume there are at most $\delta(k)^2(k+1)$ isolated vertices in $G-T$.

\section{Bounding the number of non-isolated vertices in $G-T$}
In this section we give reduction rules bounding the number of non-isolated vertices in $G-T$ by a polynomial $\phi(k)=O(k^{62})$.

Let $C$ be the set of non-isolated vertices in $G-T$.
Recall that $G[C]$ is a bipartite permutation graph. Let $A$ and $B$ be bipartition sides of $G[C]$ ordered by strong ordering $(<_A, <_B)$.
Let us define $\epsilon(k) = \delta(k)^2(k+1)+2(k+1)\delta(k)+2(k+1)^2 = O(k^{25})$.
\begin{claim}
\label{claim:deg}
Let $v\in C$ be a vertex such that $|N_C(v)|> \epsilon(k)$. Then there is an irrelevant vertex in $N_C(v)$.
\end{claim}
\begin{proof}
Assume $v\in A$ (the case of $v\in B$ is symmetrical).
Therefore, the vertices of $N_C(v)$ appear consecutively in $B$. 
Construct a set $Z'$ as follows. 
Iterate over every vertex $u$ in $T$. If $|N_C(u)\cap N_C(v)|\leq 2k+2$ then add $N_C(u)\cap N_C(v)$ to $Z'$. Otherwise, add $k+1$ leftmost (with respect to $<_B$) vertices in $N_C(u)\cap N_C(v)$ to $Z'$ and add $k+1$ rightmost (with respect to $<_B$) vertices in $N_C(u)\cap N_C(v)$ to $Z'$.

Note that at this point $|Z'|\leq 2(k+1)\delta(k)$.

Consider a vertex $v' \in N_C(v)-(Z\cup Z')$. We have that the vertices of $N_C(v')$ appear consecutively in $A$.
Denote the leftmost vertex of $N_C(v')$ by $l(v')$ and the rightmost vertex by $r(v')$.
Suppose there are $k+2$ vertices $v_1 <_B ... <_B v_{k+2}$ in $N_C(v)-(Z\cup Z')$ such that $r(v_1)=r(v_2)=...=r(v_{k+2})$ (see Figure~\ref{fig:irrelevant}).
We show that $v_{k+2}$ is irrelevant.

\begin{figure}[htp!]
\centering

\begin{tikzpicture}[xscale=1.3,yscale=2]

\begin{axis}[height=2.6cm, width=12cm,
    hide axis,
    view = {0}{90},
    at={(-1.22cm,0)}
    ]
 \addplot3 [
    surf,
    colormap={blackwhite}{gray(0cm)=(1); gray(1cm)=(0.9)},
    shader     = faceted interp,
    point meta = x,
    samples    = 15,
    samples y  = 3,
    z buffer   = sort,
    domain     = -.5:8.5,
    y domain   = 0:1
    ] (
    {x},
    {y/4},
    {0}
    );
 \addplot3 [color=black,
    domain     = -.5:8.5,samples y=0,samples=2*(640/360)*24+1,
    ] (
    {x},
    {0},
    {0} 
    );
    
    \addplot3 [color=black,
    domain     = -.5:8.5,samples y=0,samples=2*(640/360)*24+1,
    ] (
    {x},
    {1/4},
    {0}
    );   
    
\end{axis}

\coordinate (w1) at (0,1) {};
\coordinate (w2) at (1,1) {};
\coordinate (w3) at (2,1) {};
\coordinate (w4) at (3,1) {};
\coordinate (w5) at (4,1) {};
\coordinate (w6) at (5,1) {};
\coordinate (w7) at (6,1) {};
\coordinate (w8) at (7,1) {};
\coordinate (w9) at (8,1) {};

\coordinate (lw1) at (0,1.33) {};
\coordinate (lw2) at (1,1.33) {};
\coordinate (lw3) at (2,1.33) {};
\coordinate (lw4) at (3,1.33) {};
\coordinate (lw5) at (4,1.4) {};
\coordinate (lw6) at (5,1.4) {};
\coordinate (lw7) at (6,1.4) {};
\coordinate (lw8) at (7,1.33) {};
\coordinate (lw9) at (8,1.33) {};

\coordinate (u1) at (0,0) {};
\coordinate (u2) at (1,0) {};
\coordinate (u3) at (2,0) {};
\coordinate (u4) at (3,0) {};
\coordinate (u5) at (4,0) {};
\coordinate (u6) at (5,0) {};
\coordinate (u7) at (6,0) {};
\coordinate (u8) at (7,0) {};
\coordinate (u9) at (8,0) {};

\coordinate (lu1) at (0,-0.33) {};
\coordinate (lu2) at (1,-0.33) {};
\coordinate (lu3) at (2,-0.33) {};
\coordinate (lu4) at (3,-0.33) {};
\coordinate (lu5) at (4,-0.33) {};
\coordinate (lu6) at (5,-0.33) {};
\coordinate (lu7) at (6,-0.33) {};
\coordinate (lu8) at (7,-0.33) {};
\coordinate (lu9) at (8,-0.33) {};

\begin{scope}[fill opacity=0.1]
\draw[rounded corners=7, fill=gray!30, thick] (0.7,0.8)--(0.7,1.2) -- (6.3,1.2) -- (6.3,0.8)--(3,-0.05)--cycle;
\draw[rounded corners=7, fill=gray!45, thick] (0.7,0.8)--(0.7,1.2) -- (6.3,1.2) -- (6.3,0.8)--cycle;
\draw[rounded corners=7, fill=gray!30, thick] (1.7,0.8)--(1.7,1.2) -- (6.3,1.2) -- (6.3,0.8)--(4,-0.05)--cycle;
\draw[rounded corners=7, fill=gray!60, thick] (1.7,0.8)--(1.7,1.2) -- (6.3,1.2) -- (6.3,0.8)--cycle;
\draw[rounded corners=7, fill=gray!30, thick] (2.7,0.8)--(2.7,1.2) -- (6.3,1.2) -- (6.3,0.8)--(5,-0.05)--cycle;
\draw[rounded corners=7, fill=gray!85, thick] (2.7,0.8)--(2.7,1.2) -- (6.3,1.2) -- (6.3,0.8)--cycle;
\draw[rounded corners=7, fill=gray!30, thick] (3.7,0.8)--(3.7,1.2) -- (6.3,1.2) -- (6.3,0.8)--(6,-0.05)--cycle;
\draw[rounded corners=7, fill=gray!90, thick] (3.7,0.8)--(3.7,1.2) -- (6.3,1.2) -- (6.3,0.8)--cycle;

\end{scope}



\tikzstyle{every node}=[circle,minimum size=5pt,inner sep=0pt,draw,fill]

\node at (w6) {};
\node at (w7) {};

\node at (u4) {};
\node at (u5) {};
\node at (u6) {};
\node at (u7) {};

\tikzstyle{every node}=[inner sep=2pt]

\node at (lw6) {$v$};
\node at (lw7) {$r(v_4)$};

\node at (lu4) {$v_1$};
\node at (lu5) {$v_2$};
\node at (lu6) {$v_3$};
\node at (lu7) {$v_4$};

\node at (-1, -0.2) {$B$};
\node at (-1, 1.2) {$A$};

\end{tikzpicture}

\caption{\label{fig:irrelevant}
Vertices $v_1, v_2, v_3, v_4\in N_C(v)-(Z\cup Z')$ for $k=2$ such that $r(v_1)=r(v_2)=r(v_3)=r(v_4)$. We show that $v_4$ is irrelevant.
}
\end{figure}

\begin{subclaim}
\label{claim:iv}
Instance $(G, k)$ is equivalent to $(G-v_{k+2}, k)$.
\end{subclaim}
\begin{proof}
Assume that $(G-v_{k+2}, k)$ is a \yes-instance. Let $X$ be a solution for $(G-v_{k+2}, k)$. We claim that $X$ is a solution for $(G, k)$ as well. Suppose the contrary. Analogously to Claim~\ref{claim:single} we conclude that there is a long induced cycle $S$ containing $v_{k+2}$ in $G-X$.
Let $w$ and $w'$ be the neighbours of $v_{k+2}$ on $S$. 
Consider three cases:
\begin{itemize}
    \item Suppose $w,w'\in T$. Then, we conclude that there is a $5^+$-hole in $(G-v_{k+2})-X$ analogously to Claim~\ref{claim:single}.
    \item Suppose $w\notin T$ and $w'\in T$. Assume $w \leq_A v$ (the other case is symmetrical). 
    Since $v_{k+2}$ was not added to $Z'$, $Z'$ has at least $k+1$ vertices from $N_C(v) \cap N_C(w')$ to the left of $v_{k+2}$ in~$B$. 
    As $|X| \leq k$, there is a vertex $u \in Z' - X$ such that $u <_{B} v_{k+2}$ and $u \in N_C(v) \cap N_C(w')$.
    Since $w <_A v$, $u <_B v_{k+2}$, $(u,v) \in E(G)$, $(w,v_{k+2}) \in E(G)$, we have $(u,w) \in E(G)$ as $({<_A},{<_B})$ is a strong ordering.
    So, we have $u\in N_C(w)\cap N_C(w')$, which implies there is a $5^+$-hole in $(G-v_{k+2})-X$ due to the path lemma for path $S-v_{k+2}$ in $(G-v_{k+2})-X$, giving a contradiction.
    
    \item Suppose $w,w' \notin T$. Then, there is a vertex $u \in \{v_1,\ldots,v_{k+1}\} -X$ because $|X|\leq k$. 
    Note that $N_C(v_{k+2}) \subseteq N_C(u)$ because $N_C(v_{k+2}) \subseteq N_C(v_i)$ for $i\in [k+1]$. 
    This implies there is a $5^+$-hole in $(G-v_{k+2})-X$ due to the path lemma for path $S-v_{k+2}$ in $(G-v_{k+2})-X$, which is a contradiction.
\end{itemize}
\end{proof}

Therefore, we can further assume that for every $u\in A$ there are at most $k+1$ vertices $v'$ in $N_C(v)-(Z\cup Z')$ such that $r(v')=u$. Let $Z''$ be the set containing $(k+1)^2$ leftmost vertices of $N_C(v)-(Z\cup Z')$ and $(k+1)^2$ rightmost vertices of $N_C(v)-(Z\cup Z')$. We show the following claim.

\begin{subclaim}
Let $v'$ be a vertex in $N_C(v)-(Z\cup Z'\cup Z'')$. Then $v'$ is irrelevant.
\end{subclaim}
\begin{proof}
Assume that $(G-v', k)$ is a \yes-instance. Let $X$ be a solution for $(G-v', k)$. We claim that $X$ is a solution for $(G, k)$ as well. Suppose the contrary. Analogously to Subclaim~\ref{claim:iv} we conclude that there is a long induced cycle $S$ containing $v'$ in $G-X$ and the neighbours of $v'$ on $S$ both are not in $T$.

Suppose $w, w'\notin T$ are the neighbours of $v'$ on $S$. 
Assume $w <_A v$ and $w <_A w'$ (other cases are symmetrical). 
Since $v' \notin Z\cup Z' \cup Z''$, there are $k+1$ vertices $v_1,\ldots,v_{k+1}$ in $N_C(v)$ such that 
$v_i <_B v'$ and $r(v_i)\neq r(v_{i'})$ for $i\neq i'$ in $[k+1]$. 
Therefore, there is a vertex $p \in \{v_1,\ldots,v_{k+1}\}$ such that $p\in N_C(v)-X$ and $r(p)$ is also not in $X$ since $|X|\leq k$. 
As $v\in N_C(p)$, we have $v \leq_A r(p)$.
Since $w <_A r(p)$, $p <_B v'$, $(w,v') \in E(G)$, $(p,r(p)) \in E(G)$, we have $(p,w) \in E(G)$ as $({<_A},{<_B})$
is a strong ordering.
So, $w\in N_C(p)$.

If $w' \leq_A r(p)$, then $p$ is in $N_C(w')\cap N_C(w)$, which implies there is a $5^+$-hole in $(G-v')-X$ due to the path lemma for path $S-v'$ in $(G-v')-X$.

If $r(p) <_A w'$ then we know that there are $k+1$ vertices in $N_C(v)$ to the right of $v'$.
All of them are in $N_C(r(p))\cap N_C(w')$. 
Therefore, there is a vertex $p'\in (N_C(r(p))\cap N_C(w'))-X$ since $|X|\leq k$ (see Figure~\ref{fig:deg}). Since $r(p)\in N(v')$ and $S$ is an induced cycle, $r(p)$ is not on $S$. 
Since $w\in N(p)$, $w'\in N(p')$, $r(p)\in N(p)\cap N(p')$, and $|S| \geq 13$, by the path lemma for path $S-v'$ in $(G-v')-X$ we deduce there is a $5^+$-hole in $(G-v')-X$ formed by some vertices from the set $(S-v') \cup \{p,r(p),p'\}$.  
\begin{figure}[htp!]
\centering

\begin{tikzpicture}[xscale=1.3,yscale=2]

\begin{axis}[height=2.6cm, width=12cm,
    hide axis,
    view = {0}{90},
    at={(-1.22cm,0)}
    ]
 \addplot3 [
    surf,
    colormap={blackwhite}{gray(0cm)=(1); gray(1cm)=(0.9)},
    shader     = faceted interp,
    point meta = x,
    samples    = 15,
    samples y  = 3,
    z buffer   = sort,
    domain     = -.5:8.5,
    y domain   = 0:1
    ] (
    {x},
    {y/4},
    {0}
    );
 \addplot3 [color=black,
    domain     = -.5:8.5,samples y=0,samples=2*(640/360)*24+1,
    ] (
    {x},
    {0},
    {0} 
    );
    
    \addplot3 [color=black,
    domain     = -.5:8.5,samples y=0,samples=2*(640/360)*24+1,
    ] (
    {x},
    {1/4},
    {0}
    );   
    
\end{axis}

\coordinate (w1) at (0,1) {};
\coordinate (w2) at (1,1) {};
\coordinate (w3) at (2,1) {};
\coordinate (w4) at (3,1) {};
\coordinate (w5) at (4,1) {};
\coordinate (w6) at (5,1) {};
\coordinate (w7) at (6,1) {};
\coordinate (w8) at (7,1) {};
\coordinate (w9) at (8,1) {};

\coordinate (lw1) at (0,1.33) {};
\coordinate (lw2) at (1,1.33) {};
\coordinate (lw3) at (2,1.33) {};
\coordinate (lw4) at (3,1.33) {};
\coordinate (lw5) at (4,1.4) {};
\coordinate (lw6) at (5,1.4) {};
\coordinate (lw7) at (6,1.4) {};
\coordinate (lw8) at (7,1.4) {};
\coordinate (lw9) at (8,1.33) {};

\coordinate (u1) at (0,0) {};
\coordinate (u2) at (1,0) {};
\coordinate (u3) at (2,0) {};
\coordinate (u4) at (3,0) {};
\coordinate (u5) at (4,0) {};
\coordinate (u6) at (5,0) {};
\coordinate (u7) at (6,0) {};
\coordinate (u8) at (7,0) {};
\coordinate (u9) at (8,0) {};

\coordinate (lu1) at (0,-0.33) {};
\coordinate (lu2) at (1,-0.33) {};
\coordinate (lu3) at (2,-0.33) {};
\coordinate (lu4) at (3,-0.33) {};
\coordinate (lu5) at (4,-0.33) {};
\coordinate (lu6) at (5,-0.33) {};
\coordinate (lu7) at (6,-0.33) {};
\coordinate (lu8) at (7,-0.33) {};
\coordinate (lu9) at (8,-0.33) {};

\begin{scope}[fill opacity=0.2]
\draw[rounded corners=7, fill=gray!30, thick] (4,1.05)-- (5.3,0.2)-- (5.3,-0.2)--(0.7,-0.2)--(0.7,0.2)--cycle;

\draw[rounded corners=7, fill=gray!40, thick] (0.7,-0.2)--(5.3,-0.2)--(5.3,0.2)--(0.7,0.2)--cycle;

\draw[rounded corners=7, fill=gray!30, thick] (6,1.05)-- (7.3,0.2)-- (7.3,-0.2)--(2.7,-0.2)--(2.7,0.2)--cycle;

\draw[rounded corners=7, fill=gray!65, thick] (2.7,-0.2)--(7.3,-0.2)--(7.3,0.2)--(2.7,0.2)--cycle;

\draw[rounded corners=7, fill=gray!30, thick] (7,1.05)-- (8.3,0.2)-- (8.3,-0.2)--(4.7,-0.2)--(4.7,0.2)--cycle;

\draw[rounded corners=7, fill=gray!90, thick] (4.7,-0.2)--(8.3,-0.2)--(8.3,0.2)--(4.7,0.2)--cycle;

\end{scope}



\tikzstyle{every node}=[circle,minimum size=5pt,inner sep=0pt,draw,fill]
\node at (w5) {};
\node at (w6) {};
\node at (w7) {};
\node at (w8) {};

\node at (u4) {};
\node at (u6) {};
\node at (u7) {};

\tikzstyle{every node}=[inner sep=2pt]

\node at (lw6) {$v$};
\node at (lw7) {$r(p)$};
\node at (lw5) {$w$};
\node at (lw8) {$w'$};

\node at (-1, -0.2) {$B$};
\node at (-1, 1.2) {$A$};

\node at (lu4) {$p$};
\node at (lu6) {$v'$};
\node at (lu7) {$p'$};

\end{tikzpicture}

\caption{\label{fig:deg}
We have $(p,w), (p,r(p)), (p',r(p)), (p',w') \in E(G)$. We can use $p, r(p), p'$ to replace $v'$ in $S$ to obtain a new cycle, which implies existence of a $5^+$-hole due to Lemma~\ref{lem:path}.
}
\end{figure}
\end{proof}

Therefore, if $|N_C(v)|\geq \delta(k)^2(k+1)+2(k+1)\delta(k)+2(k+1)^2+1=\epsilon(k)+1$ then we can always find an irrelevant vertex in $N_C(v)$, completing the proof of Claim~\ref{claim:deg}.
\end{proof}

 Due to the above claim we can further assume that $|N_C(v)|\leq \epsilon(k)$ for every $v\in C$.

\begin{claim}
\label{claim:neighbourhoods}
Let $v$ be a vertex in $T$ and let $A$ be a side of bipartition in $G[C]$. Let \\
$v_1, \ldots , v_{|N_A(v)|}$ be the vertices of $N_A(v)$ ordered by the strong ordering of $C$. Then $|N_A(v)|\leq 2\epsilon(k)$.
\end{claim}
\begin{proof}
Suppose for the sake of contradiction that $|N_A(v)| \geq 2\epsilon(k)+1$.
Due to Claim~\ref{claim:deg} we have that $N_C(v_i)\cap N_C(v_{i+\epsilon(k)})=\emptyset$ for $i\in [|N_A(v)|-\epsilon(k)]$. Therefore, we can pick vertices $v_1, v_{\epsilon(k)+1}, v_{2\epsilon(k)+1}\in N_A(v)$ and $u_1,u_2,u_3\in C-A$ such that $(v_1,u_1) \in E(G)$, $ (v_{\epsilon(k)+1},u_2) \in E(G)$, and $(v_{2\epsilon(k)+1},u_3) \in E(G)$.
First, we argue that $(v,u_1), (v,u_2), (v,u_3) \notin E(G)$.
Indeed, otherwise we have a forbidden structure $K_3$ with only 1 vertex in $T$, which cannot be the case due to Claim~\ref{claim:one}.
However, now we have that $\{v, v_1, v_{\epsilon(k)+1}, v_{2\epsilon(k)+1}, u_1, u_2, u_3\}$ forms a forbidden structure $T_2$ with only 1 vertex in $T$, which cannot be the case due to Claim~\ref{claim:one}.
\end{proof}

Let $l=13\epsilon(k)+3$.
Suppose $|C|\geq 2((4\epsilon(k)\delta(k)+1)(l+\epsilon(k)+1))$.
Without loss of generality assume $|A|\geq |B|$. Due to Claim~\ref{claim:neighbourhoods} and adjacency property of the strong ordering, we can observe the following.
\begin{claim}
There are consecutive vertices $A_1=\{a_1, ..., a_l\}$ in $A$ such that $G[N[A_1]]$ is a bipartite permutation graph with no vertex adjacent to $T$.
\end{claim}

Let $q=5\epsilon(k)$. Then, the vertices $a_q, ..., a_{l-q+1}$ are at distance at least 10 to $T$. Let $x=a_q$ and $y=a_{l-q+1}$. 
Let $A_2=\{a_{q}, \ldots, a_{l-q+1}\}$ and $A_3=\{a_{q+\epsilon(k)+1}, \ldots, a_{l-q-\epsilon(k)}\}$.
Note that $|A_3|\geq \epsilon(k)+1$.
Let $R=N[A_2]$ and $R'=N[A_3]$ (see Figure~\ref{fig:big_component}).
\begin{figure}[htp!]
\centering

\begin{tikzpicture}[xscale=1.3,yscale=2]

\begin{axis}[height=2.6cm, width=12cm,
    hide axis,
    view = {0}{90},
    at={(-1.22cm,0)}
    ]
 \addplot3 [
    surf,
    colormap={blackwhite}{gray(0cm)=(1); gray(1cm)=(0.9)},
    shader     = faceted interp,
    point meta = x,
    samples    = 15,
    samples y  = 3,
    z buffer   = sort,
    domain     = -.5:8.5,
    y domain   = 0:1
    ] (
    {x},
    {y/4},
    {0}
    );
 \addplot3 [color=black,
    domain     = -.5:8.5,samples y=0,samples=2*(640/360)*24+1,
    ] (
    {x},
    {0},
    {0} 
    );
    
    \addplot3 [color=black,
    domain     = -.5:8.5,samples y=0,samples=2*(640/360)*24+1,
    ] (
    {x},
    {1/4},
    {0}
    );   
    
\end{axis}

\coordinate (w1) at (0,1) {};
\coordinate (w2) at (1,1) {};
\coordinate (w3) at (2,1) {};
\coordinate (w4) at (3,1) {};
\coordinate (w5) at (4,1) {};
\coordinate (w6) at (5,1) {};
\coordinate (w7) at (6,1) {};
\coordinate (w8) at (7,1) {};
\coordinate (w9) at (8,1) {};

\coordinate (lw1) at (0,1.33) {};
\coordinate (lw2) at (1,1.33) {};
\coordinate (lw3) at (2,1.33) {};
\coordinate (lw4) at (3,1.33) {};
\coordinate (lw5) at (4,1.4) {};
\coordinate (lw6) at (5,1.4) {};
\coordinate (lw7) at (6,1.4) {};
\coordinate (lw8) at (7,1.4) {};
\coordinate (lw9) at (8,1.33) {};

\coordinate (u1) at (0,0) {};
\coordinate (u2) at (1,0) {};
\coordinate (u3) at (2,0) {};
\coordinate (u4) at (3,0) {};
\coordinate (u5) at (4,0) {};
\coordinate (u6) at (5,0) {};
\coordinate (u7) at (6,0) {};
\coordinate (u8) at (7,0) {};
\coordinate (u9) at (8,0) {};

\coordinate (lu1) at (0,-0.33) {};
\coordinate (lu2) at (1,-0.33) {};
\coordinate (lu3) at (2,-0.33) {};
\coordinate (lu4) at (3,-0.33) {};
\coordinate (lu5) at (4,-0.33) {};
\coordinate (lu6) at (5,-0.33) {};
\coordinate (lu7) at (6,-0.33) {};
\coordinate (lu8) at (7,-0.33) {};
\coordinate (lu9) at (8,-0.33) {};

\begin{scope}[fill opacity=0.3]
\draw[rounded corners=7, fill=gray!30, thick] (8.3,1.2)-- (9.1,-0.2)--(-1.1,-0.2)--(-0.3,1.2)--cycle;

\draw[rounded corners=7, fill=gray!60, thick] (6.3,1.15)-- (7.3,-0.15)--(0.7,-0.15)--(1.7,1.15)--cycle;

\draw[rounded corners=7, fill=gray!90, thick] (6.3,1.1)-- (6.3,0.9)--(1.7,0.9)--(1.7,1.1)--cycle;

\end{scope}



\node at (2,0.2) {$R'$};
\node at (-0.2,0.2) {$R$};
\node at (4.3,0.7) {$A_3$};

\tikzstyle{every node}=[circle,minimum size=5pt,inner sep=0pt,draw,fill]

\node at (w1) {};
\node at (w3) {};
\node at (w4) {};
\node at (w5) {};
\node at (w6) {};
\node at (w7) {};
\node at (w9) {};


\tikzstyle{every node}=[inner sep=2pt]

\node at (lw1) {$x=a_q$};

\node at (lw3) {$a_{q+\epsilon(k)+1}$};

\node at (lw7) {$a_{l-q-\epsilon(k)}$};

\node at (lw9) {$y=a_{l-q+1}$};


\end{tikzpicture}

\caption{\label{fig:big_component}
The sets $R$, $R'$ and $A_3$.
All vertices of $R$ are at distance at least 9 to $T$.
}
\end{figure}
Let $s$ be the size of the minimum vertex cut between $x$ and $y$ in $C$. Observe that this vertex cut is contained in $R$.
Note that since $N(x)$ is a vertex cut between $x$ and $y$ we have that $s\leq \epsilon(k)$.

Let $G'$ be a graph obtained from $G$ by replacing $A_3$ with a set $S$ consisting of $s$ vertices with the same neighbourhood equal to $N(A_3)$ (see Figure~\ref{fig:reduction_kernel}).
Note that $G'$ has fewer vertices than~$G$.
We claim the following.
\begin{figure}[htp!]
\centering

\begin{tikzpicture}[xscale=1.3,yscale=2]

\begin{axis}[height=2.6cm, width=12cm,
    hide axis,
    view = {0}{90},
    at={(-1.22cm,0)}
    ]
 \addplot3 [
    surf,
    colormap={blackwhite}{gray(0cm)=(1); gray(1cm)=(0.9)},
    shader     = faceted interp,
    point meta = x,
    samples    = 15,
    samples y  = 3,
    z buffer   = sort,
    domain     = -.5:8.5,
    y domain   = 0:1
    ] (
    {x},
    {y/4},
    {0}
    );
 \addplot3 [color=black,
    domain     = -.5:8.5,samples y=0,samples=2*(640/360)*24+1,
    ] (
    {x},
    {0},
    {0} 
    );
    
    \addplot3 [color=black,
    domain     = -.5:8.5,samples y=0,samples=2*(640/360)*24+1,
    ] (
    {x},
    {1/4},
    {0}
    );   
    
\end{axis}

\coordinate (w1) at (0,1) {};
\coordinate (w2) at (1,1) {};
\coordinate (w3) at (2,1) {};
\coordinate (w4) at (3,1) {};
\coordinate (w5) at (4,1) {};
\coordinate (w6) at (5,1) {};
\coordinate (w7) at (6,1) {};
\coordinate (w8) at (7,1) {};
\coordinate (w9) at (8,1) {};

\coordinate (lw1) at (0,1.33) {};
\coordinate (lw2) at (1,1.33) {};
\coordinate (lw3) at (2,1.33) {};
\coordinate (lw4) at (3,1.33) {};
\coordinate (lw5) at (4,1.4) {};
\coordinate (lw6) at (5,1.4) {};
\coordinate (lw7) at (6,1.4) {};
\coordinate (lw8) at (7,1.4) {};
\coordinate (lw9) at (8,1.33) {};

\coordinate (u1) at (0,0) {};
\coordinate (u2) at (1,0) {};
\coordinate (u3) at (2,0) {};
\coordinate (u4) at (3,0) {};
\coordinate (u5) at (4,0) {};
\coordinate (u6) at (5,0) {};
\coordinate (u7) at (6,0) {};
\coordinate (u8) at (7,0) {};
\coordinate (u9) at (8,0) {};

\coordinate (lu1) at (0,-0.33) {};
\coordinate (lu2) at (1,-0.33) {};
\coordinate (lu3) at (2,-0.33) {};
\coordinate (lu4) at (3,-0.33) {};
\coordinate (lu5) at (4,-0.33) {};
\coordinate (lu6) at (5,-0.33) {};
\coordinate (lu7) at (6,-0.33) {};
\coordinate (lu8) at (7,-0.33) {};
\coordinate (lu9) at (8,-0.33) {};

\begin{scope}[fill opacity=0.3]
\draw[rounded corners=7, fill=gray!30, thick] (7.3,1.2)-- (9.1,-0.2)--(-1.1,-0.2)--(0.7,1.2)--cycle;

\draw[rounded corners=7, fill=gray!60, thick] (5.3,1.15)-- (7.3,-0.15)--(0.7,-0.15)--(2.7,1.15)--cycle;

\draw[rounded corners=7, fill=gray!90, thick] (5.3,1.1)-- (5.3,0.9)--(2.7,0.9)--(2.7,1.1)--cycle;

\end{scope}



\node at (4.3,0.7) {$S$};

\tikzstyle{every node}=[circle,minimum size=5pt,inner sep=0pt,draw,fill]

\node at (w2) {};
\node at (w4) {};
\node at (w5) {};
\node at (w6) {};
\node at (w8) {};


\tikzstyle{every node}=[inner sep=2pt]

\node at (lw2) {$x$};

\node at (lw8) {$y$};


\end{tikzpicture}

\caption{\label{fig:reduction_kernel}
Replacing the set $A_3$ with $S$.
Note that $|S|<|A_3|$.
}
\end{figure}

\begin{claim}
Instance $(G, k)$ is equivalent to $(G', k)$.
\end{claim}
\begin{proof}
First, assume that $(G, k)$ is a \yes-instance. Let $X$ be a minimal solution for $(G, k)$.

\begin{subclaim}
\label{subclaim:xy}
If $X\cap R' \neq \emptyset$, then $x\notin X$ and $y\notin X$.
\end{subclaim}
\begin{proof}
Suppose that $x\in X$. Let $v$ be a vertex in $X\cap R'$. Since $X$ is minimal, there is a forbidden structure $Q$ in $G$ such that $Q\cap X = \{v\}$. Since $v$ is at distance at least 9 to $T$, $Q$ must be an induced cycle of length at least 9.
Similarly, there is an induced cycle $Q'$ in $G$ such that $Q'\cap X=\{x\}$. But since $x\notin N[R']$, there is a forbidden induced cycle $Q''\subseteq Q\cup Q'$ such that $Q''\cap X=\emptyset$ due to the properties of strong ordering, which is a contradiction (see Figure~\ref{fig:combine}). The case of $y\in X$ is symmetrical.
\begin{figure}[htp!]
\centering

\begin{tikzpicture}[xscale=1.3,yscale=1.5]

\begin{axis}[height=2.6cm, width=12cm,
    hide axis,
    view = {0}{90},
    at={(-1.22cm,0)}
    ]
 \addplot3 [
    surf,
    colormap={blackwhite}{gray(0cm)=(1); gray(1cm)=(0.9)},
    shader     = faceted interp,
    point meta = x,
    samples    = 15,
    samples y  = 3,
    z buffer   = sort,
    domain     = -.5:8.5,
    y domain   = 0:1
    ] (
    {x},
    {y/4},
    {0}
    );
 \addplot3 [color=black,
    domain     = -.5:8.5,samples y=0,samples=2*(640/360)*24+1,
    ] (
    {x},
    {0},
    {0} 
    );
    
    \addplot3 [color=black,
    domain     = -.5:8.5,samples y=0,samples=2*(640/360)*24+1,
    ] (
    {x},
    {1/4},
    {0}
    );   
    
\end{axis}

\coordinate (w1) at (0,1) {};
\coordinate (w2) at (1,1) {};
\coordinate (w3) at (2,1) {};
\coordinate (w4) at (3,1) {};
\coordinate (w5) at (4,1) {};
\coordinate (w6) at (5,1) {};
\coordinate (w7) at (6,1) {};
\coordinate (w8) at (7,1) {};
\coordinate (w9) at (8,1) {};

\coordinate (lw1) at (0,1.33) {};
\coordinate (lw2) at (1,1.33) {};
\coordinate (lw3) at (2,1.33) {};
\coordinate (lw4) at (3,1.33) {};
\coordinate (lw5) at (4,1.33) {};
\coordinate (lw6) at (5,1.33) {};
\coordinate (lw7) at (6,1.33) {};
\coordinate (lw8) at (7,1.33) {};
\coordinate (lw9) at (8,1.33) {};

\coordinate (u1) at (0,0) {};
\coordinate (u2) at (1,0) {};
\coordinate (u3) at (2,0) {};
\coordinate (u4) at (3,0) {};
\coordinate (u5) at (4,0) {};
\coordinate (u6) at (5,0) {};
\coordinate (u7) at (6,0) {};
\coordinate (u8) at (7,0) {};
\coordinate (u9) at (8,0) {};

\coordinate (lu1) at (0,-0.33) {};
\coordinate (lu2) at (1,-0.33) {};
\coordinate (lu3) at (2,-0.33) {};
\coordinate (lu4) at (3,-0.33) {};
\coordinate (lu5) at (4,-0.33) {};
\coordinate (lu6) at (5,-0.33) {};
\coordinate (lu7) at (6,-0.33) {};
\coordinate (lu8) at (7,-0.33) {};
\coordinate (lu9) at (8,-0.33) {};


\path (u1) edge[thick] (w2);
\path (u3) edge[thick] (w2);
\path (u3) edge[thick] (w4);

\path (u5) edge[thick] (w4);
\path (u5) edge[thick] (w6);

\path (u7) edge[thick] (w6);
\path (u7) edge[thick] (w8);

\path (u2) edge[thick, red] (w3);
\path (u4) edge[thick, red] (w3);
\path (u4) edge[thick, red] (w5);
\path (u6) edge[thick, red] (w5);
\path (u6) edge[thick, red] (w7);
\path (u8) edge[thick, red] (w7);

\path (u3) edge[thick, dashed] (w3);
\path (u5) edge[thick, dashed] (w5);

\tikzstyle{every node}=[circle,minimum size=5pt,inner sep=0pt,draw,fill]
\node at (w2) {};
\node at (w3) {};
\node at (w4) {};
\node at (w5) {};
\node at (w6) {};
\node at (w7) {};
\node at (w8) {};
\node at (u1) {};
\node at (u2) {};
\node at (u3) {};
\node at (u4) {};
\node at (u5) {};
\node at (u6) {};
\node at (u7) {};
\node at (u8) {};

\tikzstyle{every node}=[inner sep=2pt]
\node at (lw2) {$x$};
\node at (lw3) {$q_1$};
\node at (lw4) {$q'_2$};
\node at (lw5) {$q_2$};

\node at (lu3) {$q'_1$};
\node at (lu4) {$v$};
\node at (lu5) {$q'_3$};

\end{tikzpicture}

\caption{\label{fig:combine}
The edges of $Q$ are depicted in red and the edges of $Q'$ are depicted in black. In the example above we can obtain a new cycle by taking $Q$ and replacing the path $q_1vq_2$ with $q_1q'_1q'_2q'_3q_2$, which then implies the existence of a forbidden cycle $Q''$ such that $Q''\cap X = \emptyset$.
}
\end{figure}
\end{proof}
\begin{subclaim}
$|X\cap R| \geq s$ or $X\cap R'=\emptyset$.
\end{subclaim}
\begin{proof}
Suppose the contrary. Let $v$ be a vertex in $X\cap R'$. Since $X$ is minimal, there is a forbidden structure $Q$ in $G$ such that $Q\cap X = \{v\}$. Since $v$ is at distance at least 9 to $T$, $Q$ must be an induced cycle of length at least 9. Since $x\notin X$ and $y\notin X$ due to Subclaim~\ref{subclaim:xy}, and $|X\cap R| < s$, there is an $x-y$ path in $C-X$. But since $x\notin N[R']$ and $y\notin N[R']$, there is a forbidden induced cycle $Q'$ in $G$ such that $Q'\cap X=\emptyset$ due to the properties of strong ordering, which is a contradiction (see Figure~\ref{fig:replace}).
\begin{figure}[htp!]
\centering

\begin{tikzpicture}[xscale=1.3,yscale=1.5]

\begin{axis}[height=2.6cm, width=12cm,
    hide axis,
    view = {0}{90},
    at={(-1.22cm,0)}
    ]
 \addplot3 [
    surf,
    colormap={blackwhite}{gray(0cm)=(1); gray(1cm)=(0.9)},
    shader     = faceted interp,
    point meta = x,
    samples    = 15,
    samples y  = 3,
    z buffer   = sort,
    domain     = -.5:8.5,
    y domain   = 0:1
    ] (
    {x},
    {y/4},
    {0}
    );
 \addplot3 [color=black,
    domain     = -.5:8.5,samples y=0,samples=2*(640/360)*24+1,
    ] (
    {x},
    {0},
    {0} 
    );
    
    \addplot3 [color=black,
    domain     = -.5:8.5,samples y=0,samples=2*(640/360)*24+1,
    ] (
    {x},
    {1/4},
    {0}
    );   
    
\end{axis}

\coordinate (w1) at (0,1) {};
\coordinate (w2) at (1,1) {};
\coordinate (w3) at (2,1) {};
\coordinate (w4) at (3,1) {};
\coordinate (w5) at (4,1) {};
\coordinate (w6) at (5,1) {};
\coordinate (w7) at (6,1) {};
\coordinate (w8) at (7,1) {};
\coordinate (w9) at (8,1) {};

\coordinate (lw1) at (0,1.33) {};
\coordinate (lw2) at (1,1.33) {};
\coordinate (lw3) at (2,1.33) {};
\coordinate (lw4) at (3,1.33) {};
\coordinate (lw5) at (4,1.33) {};
\coordinate (lw6) at (5,1.33) {};
\coordinate (lw7) at (6,1.33) {};
\coordinate (lw8) at (7,1.33) {};
\coordinate (lw9) at (8,1.33) {};

\coordinate (u1) at (0,0) {};
\coordinate (u2) at (1,0) {};
\coordinate (u3) at (2,0) {};
\coordinate (u4) at (3,0) {};
\coordinate (u5) at (4,0) {};
\coordinate (u6) at (5,0) {};
\coordinate (u7) at (6,0) {};
\coordinate (u8) at (7,0) {};
\coordinate (u9) at (8,0) {};

\coordinate (lu1) at (0,-0.33) {};
\coordinate (lu2) at (1,-0.33) {};
\coordinate (lu3) at (2,-0.33) {};
\coordinate (lu4) at (3,-0.33) {};
\coordinate (lu5) at (4,-0.33) {};
\coordinate (lu6) at (5,-0.33) {};
\coordinate (lu7) at (6,-0.33) {};
\coordinate (lu8) at (7,-0.33) {};
\coordinate (lu9) at (8,-0.33) {};


\path (u3) edge[thick] (w2);
\path (u3) edge[thick] (w4);

\path (u5) edge[thick] (w4);
\path (u5) edge[thick] (w6);

\path (u7) edge[thick] (w6);
\path (u7) edge[thick] (w8);

\path (u2) edge[thick, red] (w1);
\path (u2) edge[thick, red] (w3);
\path (u4) edge[thick, red] (w3);
\path (u4) edge[thick, red] (w5);
\path (u6) edge[thick, red] (w5);
\path (u6) edge[thick, red] (w7);
\path (u8) edge[thick, red] (w7);
\path (u8) edge[thick, red] (w9);

\path (u2) edge[thick, dashed] (w2);
\path (u8) edge[thick, dashed] (w8);

\tikzstyle{every node}=[circle,minimum size=5pt,inner sep=0pt,draw,fill]
\node at (w1) {};
\node at (w2) {};
\node at (w3) {};
\node at (w4) {};
\node at (w5) {};
\node at (w6) {};
\node at (w7) {};
\node at (w8) {};
\node at (w9) {};
\node at (u2) {};
\node at (u3) {};
\node at (u4) {};
\node at (u5) {};
\node at (u6) {};
\node at (u7) {};
\node at (u8) {};

\tikzstyle{every node}=[inner sep=2pt]
\node at (lw2) {$x$};
\node at (lw5) {$v$};
\node at (lw8) {$y$};

\node at (lu2) {$q_1$};
\node at (lu8) {$q_2$};

\end{tikzpicture}

\caption{\label{fig:replace}
The edges of $Q$ are depicted in red and the edges of $x-y$ path are depicted in black. In the example above we can obtain a new cycle by taking $Q$ and replacing the inner vertices of $q_1-q_2$ path with the $x-y$ path, which then implies the existence of a forbidden cycle $Q'$ such that $Q'\cap X = \emptyset$.
}
\end{figure}
\end{proof}

\begin{subclaim}
If $X\cap R'=\emptyset$, then $X$ is a solution for $(G', k)$.
\end{subclaim}
\begin{proof}
Suppose there is a forbidden structure $Q$ in $G'-X$. Since $X$ is a solution for $(G, k)$, $Q$~must be an induced cycle of length at least 9 passing through $S$. But since for every $v\in S$ we have $N(v)=N(A_3)$ and every two vertices $a,a'\in A_3$ are connected in $G[R']$ (because there is a path between $x$ and $y$ in $G[C]$ as $|S|>0$), we have that there is a forbidden induced cycle $Q'$ in $G-X$ due to the properties of strong ordering, which is a contradiction.
\end{proof}

\begin{subclaim}
If $|X\cap R|\geq s$, then the set $X'=(X-R)\cup S$ is a solution for $(G', k)$.
\end{subclaim}
\begin{proof}
Since $|S|=s$, we have $|X'|\leq |X|\leq k$.

Suppose there is a forbidden structure $Q$ in $G'-X'$. 
Since $S \subseteq X'$, $Q$ is also a forbidden structure in $G - A_3$.
Since $X$ is a solution for $(G, k)$, $Q$~must
intersect the set $R -S$, and hence $Q$ is an induced cycle of size at least 9.
Let $u$ and $w$ be the vertices from $Q \cap A$ closest to $A_3$ from the left and the right side, respecitvely.
In particular, due to the properties of the strong ordering $({<_A},{<_B})$,
$u$ and $w$ must have a common neighbour in $B$.
On the other hand, since $|A_3|\geq \epsilon(k)+1$, we have $N_B(u) \cap N_B(w) = \emptyset$,
which is a contradiction.
\end{proof}

For the second direction, assume $(G', k)$ is a \yes-instance. Let $Y$ be a minimal solution for $(G', k)$.
\begin{subclaim}
If $Y\cap S=\emptyset$, then $Y$ is a solution for $(G, k)$.
\end{subclaim}
\begin{proof}
Suppose there is a forbidden structure $Q$ in $G-Y$. Since $Y$ is a solution for $(G', k)$, $Q$ must be an induced cycle of length at least 9 passing through $A_3$. But since for every $v\in S$ we have $N(v)=N(A_3)$, there is a forbidden induced cycle in $G'-Y$, which is a contradiction.
\end{proof}
\begin{subclaim}
If $Y\cap S\neq \emptyset$, then $S\subseteq Y$.
\end{subclaim}
\begin{proof}
Suppose there is a vertex $u$ in $S-Y$. Let $v$ be a vertex in $Y\cap S$. Since $Y$ is minimal, there is a forbidden structure $Q$ in $G'$ such that $Q\cap Y=\{v\}$. Since $v$ is at distance at least 9 to $T$, $Q$ must be an induced cycle of length at least 9. But since $N(u)=N(v)=N(A_3)$, there is a forbidden induced cycle $Q'$ in $G'-Y$, which is a contradiction.
\end{proof}
\begin{subclaim}
Let $P$ be the minimum $(x,y)$-separator in $C$. If $Y\cap S\neq \emptyset$, then the set $Y'=(Y-S)\cup P$ is a solution for $(G, k)$.
\end{subclaim}
\begin{proof}
Since $|S|=s$ and $|P|=s$, we have $|Y'|=|Y|\leq k$.

Suppose there is a forbidden structure $Q$ in $G-Y'$. Since $Y$ is a solution for $(G', k)$, $Q$ must be an induced cycle of length at least 9 passing through $A_3$. Since $A_3$ is between $x$ and $y$ in $C$, there is an $x-y$ path in $G-Y'$ due to the properties of strong ordering, which is a contradiction because $P\subseteq Y'$.
\end{proof}

Therefore, instances $(G, k)$ and $(G, k')$ are equivalent.
\end{proof}
Let us define $\phi(k) = 2((4\epsilon(k)\delta(k)+1)(14\epsilon(k)+4)) = O(k^{62})$.
We conclude that the number of non-isolated vertices in $G-T$ is at most $\phi(k)$ or we can transform the instance $(G, k)$ into an equivalent instance with fewer vertices.

\section{The final analysis of the kernel}
Following the claims in previous section, we can finally prove the Theorem~\ref{theo:kernel}. Recall that the functions $\delta(k)$, $\epsilon(k)$ and $\phi(k)$ are polynomial in $k$. If we cannot apply any reduction rule to $G$, the following conditions hold:
\begin{itemize}
    \item $|T|\leq \delta(k)$,
    \item There are at most $\delta(k)^2(k+1)$ isolated vertices in $G-T$,
    \item There are at most $\phi(k)$ non-isolated vertices in $G-T$.
\end{itemize}
Let us define $\xi(k)=\delta(k) + \delta(k)^2(k+1) + \phi(k)$.
Due to the above, $G$ has at most $\xi(k) = O(k^{62})$ vertices.

Therefore, we conclude that either the size of $G$ is polynomially bounded in $k$ or we can repeatedly apply reduction rules in polynomial time, progressively reducing the number of vertices of~$G$.

\section{Acknowledgements}
We would like to acknowledge that a polynomial kernel with $O(k^{229})$ vertices for Vertex Deletion into Bipartite Permutation Graphs has been independently obtained by Kanesh, Madathil, Sahu, Saurabh and Verma using similar approach \cite{kanesh_et_al:LIPIcs.IPEC.2021.23}, published 22.11.2021.

The author of this paper would like to thank Tomasz Krawczyk for helpful remarks.

\newpage
\bibliographystyle{plain}
\bibliography{references}

\end{document}